\documentclass[%
 reprint,
 amsmath,amssymb,
 aps,
superscriptaddress
]{revtex4-1}

\usepackage{graphicx}
\usepackage{dcolumn}
\usepackage{bm}
\usepackage{multirow}

\usepackage{amsmath,amssymb,amsfonts}
\usepackage{graphicx}
\usepackage{color}
\usepackage{appendix}
\usepackage{amsthm}
\usepackage{tikz}
\usepackage{youngtab}

\newtheorem{theorem}{Theorem}

\newtheorem{definition}{Definition}

\newcommand{\nn}{\nonumber}
\newcommand{\bn}{\begin{enumerate}}
\newcommand{\en}{\end{enumerate}}

\newcommand{\eq}[1]{(\ref{#1})}

\def\a{\alpha}

\def\eps{\epsilon}

\def\s{\sigma}

\def\p{\psi}

\def\iff{\Longleftrightarrow}

\def\<{\langle}
\def\>{\rangle}

\def\jmath{{j}}

\def\cT{{\mathcal{T}}}
\def\cO{{\mathcal{O}}} 
\def\cE{{\mathcal{E}}} 
\def\vn{{\vec{n}}}
\def\vm{{\vec{m}}}
\def\ha{\hat{a}}
\def\hb{\hat{b}}
\def\hU{\hat{U}}

\begin{document}

\title{Majorization and the time complexity of linear optical networks}

\author{Seungbeom Chin }
\affiliation{Molecular Quantum Dynamics and Information Theory Laboratory, Department of Chemistry, Sungkyunkwan University, Suwon 16419, Korea}
\author{Joonsuk Huh}%
\email{joonsukhuh@gmail.com}
\affiliation{Molecular Quantum Dynamics and Information Theory Laboratory, Department of Chemistry, Sungkyunkwan University, Suwon 16419, Korea}


\begin{abstract}
This work shows that the majorization of photon distributions is related to  the runtime of classically simulating multimode passive linear optics, which explains one aspect of the boson sampling hardness. 
A Shur-concave quantity which we name the \emph{Boltzmann entropy of elementary quantum complexity} ($S_B^q$) is introduced to present some quantitative analysis of the relation between the majorization and the classical runtime for simulating linear optics. We compare $S_B^q$ with two quantities that are important criteria for understanding the computational cost of the photon scattering process, $\cT$ (the runtime for the classical simulation of linear optics) and $\cE$ (the additive error bound for an approximated amplitude estimator).
First, for all the known algorithms for computing the permanents of matrices with repeated rows and columns, the runtime $\cT$  becomes shorter as the input and output distribution vectors are more majorized. Second, the error bound $\cE$ decreases as the majorization difference of input and output states increases. We expect that our current results would help in understanding the feature of linear optical networks from the perspective of quantum computation. \\ 
\\
keywords: majorization, entropy, passive linear optics, indistinguishability, boson sampling
\end{abstract}

\keywords{boson sampling, majorization, entropy, particle indistinguishability}
\maketitle

\section{Introduction}


The linear optical network (LON) is one of the convincing quantum computing processors \cite{nielson2011, scully1997, kok2007} (for the explanation on the general formalism of LON, see, e.g., Refs. \cite{zukowski1997realizable, lim2005multiphoton}). Knill, Laflamme and Milburn (KLM) proved that  linear optics with post-selected measurements is sufficiently powerful to perform universal quantum computations \cite{KLM}. In addition to the theoretical motivation, LON also has many practical advantages: the high robustness of photons against decoherence \cite{tillmann2015,biggerstaff2016,xu2016}, relatively simple operations consisting of beamsplitters and phase shifters \cite{reck1994}, a routine construction of circuits in photonic chips \cite{crespi2011,shadbolt2012,carolan}. 
 
Boson sampling (BS) \cite{Aaronson2011} is a recently introduced non-universal quantum computing model based on the LON implementation.
It has drawn wide attention since it has a potental to disprove the extended Church-Turing thesis with LON \cite{broome2013, shchesnovich2013, shen2014, tichy2014s}. Motivated by Sheel's observation \cite{Scheel2004} that the scattering amplitude ot LON is proportional to matrix permanents (solving which is \# P-complete),
Aaronson and Arkhipov~\cite{Aaronson2011} demontrated that the unitary operation of a passive interferometer with single-photon input modes is not likely to be simulated efficiently by classical computers for sufficiently large systems.

 
There exist several conditions assumed to be fulfilled for the computational hardness of BS: low photon density, the preparation of identical photons, and random scattering matrix. By varying these conditions independently, we can understand their operational mechanism and what happens when they are disrupted in LON. Among the conditions, the \emph{low photon density condition} is needed for photons not to collide at the same mode (unbunched single-photon source). One way to comprehend the physical implication of this condition is to disrupt this restriction and see how the change affects the classical simulation processes of LON scattering amplitude. Therefore, it is an interesting attempt to analyze the behavior of some crucial computational quantities under the change of photon distributions (with $U$ fixed).
 
The concept of \emph{majorization} provides a rigorous mathematical tool to examine the quantitative relations. Majorization is a preorder that determines whether a real number vector is more disordered than another real number vector (for a review, see, e.g., Ref.~\cite{olkin2}). In quantum information theory, it is a useful property to compare the amount of quantum resource between two systems.
As a specific example,
Nielsen's theorem \cite{nielsen1999} in the entanglement resource theory proves that majorization is a criterion for possible LOCC (local operations and classical communication) transformations between bipartite entanglement states. There exists a counterpart theorem in the coherence resource theory \cite{du2015}, in which the majorization is a criterion for possible incoherent operations between pure states. In addition, it is conjectured in \cite{latorre2002,orus2004} that all optimal quantum algorithms experience a monotonic decrease of majorization during the tasks, which was recently observed experimentally for some algorithms including quantum Fourier transforms \cite{flamini2016}.

For our case, the photon distribution vectors of input/output modes are evaluated from the viewpoint of majorization. Our result shows that \emph{one can find  classical algorithms with more time-efficiency when the photon distribution vectors are more majorized.} 
As a measure for the majorization of systems, we introduce \emph{Boltzmann entropy of elementary quantum complexity} $S_B^q$ (a Schur-concave function). From the close relation of $S_B$ with the amount of $\cT$ (the classical runtime for the exact computation) and  $\cE$ (the additive error bound for an approximated permanent estimator) for a LON, we expect that $S_B^q$ plays an important role to understand the quantum supremacy of BS \footnote{It is also interesting to compare this type of entropy with the generalized concurrence of Fock state introduced in \cite{huh2016computing}.}. This relation  also suggests that the potential quantum supremacy of BS arises not necessarily from the entanglement. For, while the amount of entanglement depends on the partition of modes into different observers (for the generation of entanglement in LON, see, e.g., Ref.~\cite{stanisic2017generating}), the $\cT$ and $\cE$ we can find seems to depend purely on the input/output photon distribution patterns.


 

This paper is organized as follows: 
 In Section \ref{majorization}, we present some physical and mathematical review. The passive LON setup (BS) and its implication of computational complexity are briefly introduced, and then the concept of majorization and Schur concavity is explained. In Section \ref{Q}, the Boltzmann entropy $S_B$ of elementary quantum complexity is defined, which measures the inherent quantum complexity of a given input photon distribution. 
 In Section \ref{Rt}, we show that a more majorized input/output distribution implies a shorter classical runtime for the exact  simulation for a LON process.  Our observation fits all the known generalized algorithms for computing transition amplitudes of LON.
 In Section \ref{e}, we focus on the approximated simulation, showing that the runtime and universal additive error bound for the approximated algorithm is interpreted with the majorization and $S_B$. In Section \ref{discussion}, we summarize our results and discuss the future direction that can extend the current research.

\section{Physical and Mathematical preliminaries}\label{majorization}
In this section, we review fundamental backgrounds for our discussion, which are the computational complexity problem of the passive LON scattering process and the concept of majorization.

\subsection*{The computational complexity of passive LON scattering}

 \begin{figure}{t}
 	\centering
 	\begin{tikzpicture}
 	\draw (-0.6,0.2) node[left] {$|n_1\>$} ;
 	\draw (-0.6,0.6) node[left] {$|n_2\>$} ;
 	\draw (-0.6,1) node[left] {$|n_3\>$} ;
 	\draw (-0.5,2.6) node[left] {$|n_M\>$} ;
 	\draw[thick, dotted] (-1,1.4) -- (-1,2.2);
 	
 	
 	\draw[thick, ->] (0,0.2) -- (0.5,0.2) ;
 	\draw[thick, ->] (0,0.6) -- (0.5,0.6) ;
 	\draw[thick, ->] (0,1) -- (0.5,1) ;
 	\draw[thick, ->] (0,1.4) -- (0.5,1.4) ;
 	\draw[thick, ->] (0,1.8) -- (0.5,1.8) ;
 	\draw[thick, ->] (0,2.2) -- (0.5,2.2) ;
 	\draw[thick, ->] (0,2.6) -- (0.5,2.6) ;

\draw[thick, ->] (0,0.2) -- (3.9,0.2) ;
\draw[thick, ->] (3.5,0.6) -- (3.9,0.6) ;
\draw[thick, ->] (0,1) -- (3.9,1) ;
\draw[thick, ->] (0,1.4) -- (3.9,1.4) ;
\draw[thick, ->] (0,1.8) -- (3.9,1.8) ;
\draw[thick, ->] (0,2.2) -- (3.9,2.2) ;
\draw[thick, ->] (0,2.6) -- (3.9,2.6) ; 	
 	
 	\draw (5.5,0.2) node[left] {$|m_1\>$} ;
 	\draw (5.5,0.6) node[left] {$|m_2\>$} ;
 	\draw (5.5,1) node[left] {$|m_3\>$} ;
 	\draw (5.6,2.6) node[left] {$|m_M\>$} ;
 	\draw[thick, dotted] (5,1.4) -- (5,2.2);
 	
 	\filldraw[fill = cyan] (0,0.2)circle [radius=0.15];
 	\filldraw[fill = cyan] (0,0.6)circle [radius=0.15];
 	\filldraw[fill = cyan] (0,1)circle [radius=0.15];
 	\filldraw[fill = cyan] (0,1.4)circle [radius=0.15];
 	\filldraw[fill = cyan] (0,1.8)circle [radius=0.15];
 	\filldraw[fill = cyan] (0,2.2)circle [radius=0.15];
 	\filldraw[fill = cyan] (0,2.6)circle [radius=0.15];
 	
 	\filldraw[fill = brown] (0.5,0) rectangle (3.5,2.8) ;
 	\draw (1.7, 1.5)  node[right] {\textbf{U}} ; 
 	\draw (0.9, 1)  node[right] {linear optical } ;
 	\draw (1.3, 0.7)  node[right] {network } ;
 	
 	\filldraw[fill = purple] (4,0.2)circle [radius=0.15]; 
 	\filldraw[fill = purple] (4,0.6)circle [radius=0.15];
 	\filldraw[fill = purple] (4,1)circle [radius=0.15];
 	\filldraw[fill = purple] (4,1.4)circle [radius=0.15];
 	\filldraw[fill = purple] (4,1.8)circle [radius=0.15];
 	\filldraw[fill = purple] (4,2.2)circle [radius=0.15];
 	\filldraw[fill = purple] (4,2.6)circle [radius=0.15];
 	\end{tikzpicture}
 	\caption{Multimode linear optical network with $N$ photons and $M$ modes, with arbitrary input/output photon distribution vectors $|\vec{n}\> =|n_1,n_2,\cdots ,n_M\>$ and $|\vec{m}\> =|m_1,m_2,\cdots ,m_M\>$. The linear optical operation $U$ is unitary, consisting of phase-shifters and beamsplitters. When $n_i \le 1$ and $m_i \le 1$ for all $i =1,\dots M$, this process is hard to simulate with classical computers \cite{Aaronson2011}.}
 	
 \end{figure}
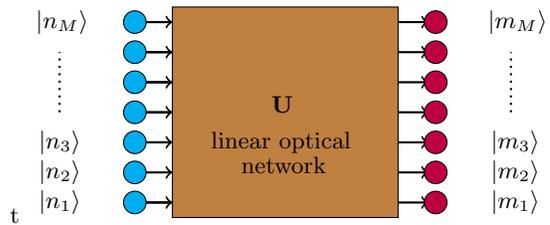

A passive LON scattering apparatus consists of beamsplitters, phase-shifters, and detectors (Fig. 1). $N$ photons are prepared for the initial state, which are places in $M$ modes of the apparatus. The input state is represented as 
\begin{align}
|\vn\> &= |n_1,n_2,\cdots, n_M\> \nn \\
        &= \frac{(\hat{a}_1^{\dagger})^{n_1} (\hat{a}_2^{\dagger})^{n_2}\cdots  (\hat{a}_M^{\dagger})^{n_M} }{\sqrt{\prod_{i=1}^{M} n_i !}}|\vec{0}\>
\end{align}
with $\sum_{i=1}^M n_i=N$ (note that $0\le n_i\le N$). We name the vector $\vn$ as photon distribution vector. The  beamsplitters and phase-shifters, which are reprsented with the unitary operator $\hat{U}$, scatter the photon distribution vector by the following map, 
\begin{align}
\hat{U}\hat{a}_i^{\dagger}\hat{U}^{\dagger} = \sum_{j=1}^M U_{ij}\hat{a}^{\dagger}_j,
\end{align} where $U_{ij}$ are the elements of a unitary matrix. Denoting the output state as $|\vm\> = |m_1,m_2,\cdots, m_M\>$, the scattering amplitude from $|\vn\> $ to $|\vm\>$ is expressed as
\begin{align}\label{perm}
\<\vm|\hat{U}|\vn\> = \frac{\textrm{Per}([U]_{\vn,\vm})}{\sqrt{\prod_{i=1}^{M}n_i!m_i!}},
\end{align}
where $\textrm{Per}(\cdot)$ means the matrix permanent, and $[U]_{\vn,\vm}$ is an $N \times N$ submatrix of $U$ ($n_i$ of the $i$-th row of $U$ and $m_j$ of $j$-th row of $U$ constitute $[U]_{\vn,\vm}$) \cite{Scheel2004}. Then the goal of BS is to sample from the transtion probability $|\<\vm|\hat{U}|\vn\>|^2$ for the given input condition  with $0\le n_i \le 1$. 

The connection between the scattering amplitude $\<\vm|\hU|\vn\>$ and matrix permanent (Eq.~\eqref{perm}) reveals crucial implications for the simulation of BS.
The computational complexity of matrix permanents is \#P-complete \cite{valiant1979complexity, Aaronson2011b}, which is considered a classically hard problem. One can expect that the computational complexity of matrix permanents is closely related to that of simulating BS. 
The efficient classical simulation of the scattering amplitude is a sufficient (but not necessary) condition for that of the transition probability (for a more detailed explanation on this point, see 2.2 of Ref. \cite{olson2018role}). It is demonstrated in Ref. \cite{Aaronson2011} that the exact BS problem is not efficiently solvable by classical Turing machines, unless P$^{\textrm{\#P}}$ = BPP$^{\textrm{NP}}$ (the polynomial hierarchy collapses to the third level, which seems unlikely for our current knowledge). This theorem implies that BS is a non-universal quantum computer that contains quantum supremacy with a relatively simple physical implementation.

The original BS setup discussed so far has the restriction that the photons are unbunched at the input modes ($0\le n_i \le 1$). On the other hand, when some $n_i$ are larger than $1$ (arbitrary photon distribution vector), the scattering process is equivalent to a submatrix permanent (Eq. \eqref{perm}), which would be simpler to simulate with classical Turin machines than the original case. The quantitative analysis about ``how simpler the simulation becomes according to the photon distribution'' is possible by focusing on the majorization pattern of the photon distribution vectors.

\subsection*{Majorization and Schur concavity}

Now we briefly introduce the definitions and properties of the majorization theory.
\begin{definition}
	(Majorization \cite{olkin2}) Let $\vec{x}= (x_1,x_2 \cdots , x_d)$ and $\vec{y}= (y_1,y_2 \cdots , y_d)$ be  nonincreasing sequences of positive real numbers. Then $\vec{x}$ is majorized by $\vec{y}$ (which is expressed as $\vec{x}$ $\prec$ $\vec{y}$) if the following relations are satisfied:
	\begin{align}
	\forall k \quad(1\le k <d):\quad   \sum_{i=1}^{k}x_i \le \sum_{i=1}^k y_i
	\end{align}
	and
	\begin{align}
	 \sum_{i=1}^{d}x_i = \sum_{i=1}^d y_i.
	\end{align}
\end{definition}

\begin{definition}
	(Discrete majorization \cite{ruch, sauerbrei}) Discrete majorization is the majorization relation between vectors with positive integers. This relation is determined by way of partitioning a given integer $N$.  Each partition with a nonincreasing sequance of non-negative integers is represented by a Young diagram (see appendix \ref{young}). 
\end{definition}
A distinctive feature of discrete majorization in contrast to continuous majorization is that the number of partitioning a given integer $N$ is finite. Therefore, we can construct a complete array of distribution vectors along the majorization with $N$ (more specific explanation and examples are given in Appendix \ref{young}). 
Since this research deals with the distribution of $N$ photons in $M$ modes, our discussion is limited to discrete majorization.

\begin{definition}
(Schur-convex and Schur-concave functions \cite{olkin2})
	A real-valued function $f$ of a positive real-valued $n$-dimensional vector
	is said to be Schur-convex if
	\begin{align}
	\vec{x} \prec \vec{y} \Rightarrow f(\vec{x}) \le f(\vec{y}). 
	\end{align}
	It is strictly Schur-convex if $f(\vec{x}) < f(\vec{y})$ whenever $\vec{x} \prec \vec{y}$.
	
	Similarly, $f$
	is Schur-concave  if
	\begin{align}
	\vec{x} \prec \vec{y} \Rightarrow f(\vec{x}) \ge f(\vec{y}) 
	\end{align}holds.
\end{definition}

Note that the inverse is not true, i.e., $f(\vec{x}) \le f(\vec{y})$ for a Schur convex function does not guarantee $\vec{x} \prec \vec{y}$. \\
\\
We enumerate some Schur-concave functions that are important for the later discussion: 
\begin{align}
\label{concave}
&X_k(\vec{n}) = \sum_{i_1 < i_2 < \cdots < i_k=1}^{d} n_{i_1}\cdots n_{i_k} \nn \\
& \qquad \textrm{(the elementary symmetric function, } 1\le k \le d) \nn \\
&\a_{\vec{n} } = \textrm{(the number of nonzero elements of } \vec{n} ), \nn \\
&Q_{\vec{n}}=\frac{N!}{\prod_{i=1}^d n_i!}, \qquad v_{\vec{n}} = \prod_{i=1}^{d}\sqrt{\frac{n_{i}!}{n_i^{n_i}} }, \nn \\
&H(\vec{n})=-\sum_{i=1}^{d}\frac{n_i}{N} \log_2 \frac{n_i}{N}.
\end{align}
Among these, $Q_{\vec{n}}$ and $H(\vec{n})$ are strictly Schur-concave. The functions in Eq.~\eqref{concave} turn out to be components of  $\cT$ (the runtime of the generalized classical algorithm for calculating the permanent) and $\cE$ (the additive error bound for the approximated permanent estimator), which measure the computing power of boson sampling systems.

\section{elementary quantum complexity and majorization}\label{Q}

It is conjectured that the indistinguishability of photons is responsible for the computational complexity of linear optics. Aaronson and Arkhipov argued in Section 1.1 of \cite{Aaronson2011} that the exchange symmetry of identical bosons creates an \emph{effective entanglement} (a kind of artificial entanglement), which would be the origin of the computational complexity in linear optics. On the other hand, when all the particles are distinguishable, the Monte Carlo method that counts each particle independently in time is efficient enough to simulate the sampling process \cite{tichy2014interf,jerrum2004}. This analysis supports the assumption that indistinguishability is the quantum computing resource of linear optics.

Indistinguishability is an intrinsic quantum property that imposes a sharp distinction between quantum and classical particles. 
Meanwhile, we should consider the possibility of a collective characteristic that determines the actual emergence of the indistinguishability effect in the entire system.   
For example, Killoran et al.~\cite{killoran2014extracting} pointed out that the identical particles in LON can create entanglement, and the amount of the entanglement depends on the distribution pattern of particles in modes. In other words, the effective entanglement mentioned in Ref.~\cite{Aaronson2011} caused the extracted entanglement in Ref.~\cite{killoran2014extracting}. From the viewpoint of the computational complexity, one can expect that the distribution pattern of photons in modes is one of the decisive factors for the computational cost of simulating the passive LON scattering processes.

In this section, we develop this idea using the majorization of particle distributions in multi-modes.
Our discussion here is restricted to the input states, which will extend to the entire system including unitary operations of the interferometer in the later sections. The analysis includes the introduction of physical quantities, elementary quantum complexity and the Boltzmann entropy of elementary quantum complexity.


\subsection*{Partition of quantum/classical particles and elementary quantum complexity $Q_{\vec{n}}$}

We first compare the possible distribution number of $N$ identical particles with that of $N$ distinguishable particles, and show that majorization determines the ratio between them. We name the ratio \emph{the elementary quantum complexity}.

Partitioning two balls into two boxes (Fig. \ref{even}) is the simplest example for our discussion.  In the classical case, the balls and boxes are all distinguishable. For the case of identical particles (LON for our setup), boxes (modes) are distinguishable  but balls (identical particles) are indistinguishable. Putting each ball into a different
box, we have two possible ways of allocation in the distinguishable setup, but only one way in the indistinguishable setup, i.e.,
\begin{align}
\frac{\textrm{(Number of allocations for distinguishable balls)} }{\textrm{(Number of allocations for indistinguishable balls) } } =2.
\end{align} 

\begin{figure}
	\centering
	\begin{tikzpicture}
	\draw (0,1) rectangle (1,2);

	\draw  (2,1) rectangle (3,2);
	\draw (0,0) rectangle (1,-1);
	\draw  (2,0) rectangle (3,-1);
	\filldraw[fill = orange] (0.5,1.5) circle [radius=0.1];
	\filldraw[fill = cyan] (2.5,1.5) circle [radius=0.1];
	\filldraw[fill = orange] (2.5,-0.5) circle [radius=0.1];
	\filldraw[fill = cyan] (0.5,-0.5) circle [radius=0.1];
	
	\draw[very thick] (4,2.5) -- (4,-1);
\draw  (-0.2,0.5) -- (3.2,0.5);
	\draw (0.5,2.5)  node[right] {Distinguishable}; 
	\draw (5,2.5)  node[right] {Indistinguishable}; 	
	\draw (5,0) rectangle (6,1);
	\draw  (7,0) rectangle (8,1);
	\filldraw[fill = black] (7.5,0.5) circle [radius=0.1];
	\filldraw[fill = black] (5.5,0.5) circle [radius=0.1];
	\end{tikzpicture}
	\caption{Two balls evenly distributed in two boxes (the left is distinguishable and the right is indistinguishable).}
	\label{even}
\end{figure}
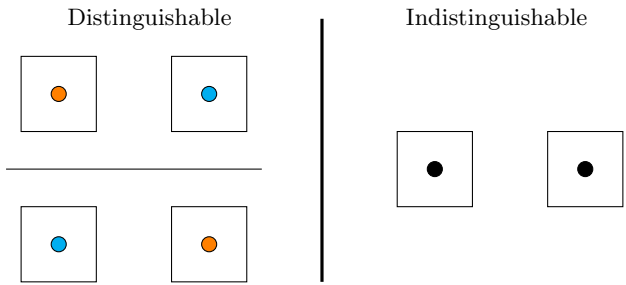

However, this difference disappears when two balls are in the same box (Fig. \ref{odd}). Now both distinguishable and indistinguishable case have two possibilities, thus we have
\begin{align}
\frac{\textrm{(Number of allocations for distinguishable balls)} }{\textrm{(Number of allocations for indistinguishable balls) } } =1.
\end{align}
Therefore, we can say that when two identical particles are unevenly distributed they have no difference with distinguishable particles from the viewpoint of number counting.
We define the ratio as $Q_{\vec{n}}$, where $\vec{n}$ represents the distribution vector for a partition, where $n_i$ is the number of particles in the $i$th box. $\vec{n}$ is given by  $(1,1)$ in Fig. \ref{even} and $(2,0)$ in Fig. \ref{odd}.
Using the concept of majorization, we have
\begin{align}
\label{q2concave}
\vec{n} \prec \vec{m} \Longrightarrow Q_{\vec{n}} > Q_{\vec{m}},
\end{align}
with $\vec{n}=(1,1)$ and $\vec{m}=(2,0)$.

\begin{figure}
	\centering
	\begin{tikzpicture}
	\draw (0,1) rectangle (1,2);
	\draw  (2,1) rectangle (3,2);
	\draw (0,0) rectangle (1,-1);
	\draw  (2,0) rectangle (3,-1);
	\filldraw[fill = orange] (0.3,1.5) circle [radius=0.1];
	\filldraw[fill = cyan] (0.7,1.5) circle [radius=0.1];
	\filldraw[fill = orange] (2.3,-0.5) circle [radius=0.1];
	\filldraw[fill = cyan] (2.7,-0.5) circle [radius=0.1];
	
		\draw[very thick] (4,2.5) -- (4,-1);
	    \draw  (-0.2,0.5) -- (3.2,0.5);
	    \draw (4.8,0.5) -- (8.2,0.5);

	\draw (0.5,2.5)  node[right] {Distinguishable}; 
\draw (5,2.5)  node[right] {Indistinguishable}; 	
	\draw (5,1) rectangle (6,2);
	\draw  (7,1) rectangle (8,2);
	\filldraw[fill = black] (5.3,1.5) circle [radius=0.1];
	\filldraw[fill = black] (5.7,1.5) circle [radius=0.1];
	\draw (5,0) rectangle (6,-1);
	\draw  (7,0) rectangle (8,-1);
	\filldraw[fill = black] (7.3,-0.5) circle [radius=0.1];
	\filldraw[fill = black] (7.7,-0.5) circle [radius=0.1];
	\end{tikzpicture}
	\caption{Two balls unevenly distributed in two boxes (the left is distinguishable and the right is indistinguishable).}
	\label{odd}
\end{figure}
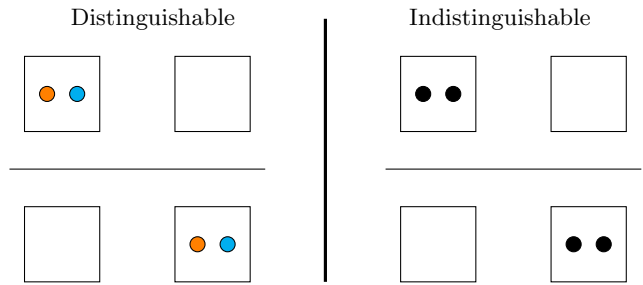

 The argument for 2 balls in 2 boxes can be generalized to the case of $N$ arbitrary balls in $M$ different boxes in a similar manner. We can quantify $Q_{\vec{n}}$ in general for which Eq. \eqref{q2concave} still holds.

\begin{definition}
The elementary quantum complexity $Q_{\vn}$, which evaluates the quantum aspect of identical particle systems from the viewpoint of the partition arrangement, is defined as
\begin{align}
 Q_{\vn} = \frac{N!}{\prod_{i}(n_i)!}. \nn
\end{align}

\end{definition}
Then the following theorem holds:
\begin{theorem}
	\label{1}
	For a distribution vector $\vec{n}$ of $N$ indistinguishable particles, there corresponds to $Q_{\vn}$ ways of  arranging $N$ distinguishable particles. If $\vec{n} \prec \vec{m}$, then $Q_{\vec{n} } $ $>$ $Q_{\vec{m}}$.  
\end{theorem}
\begin{proof}
	Let us consider one vector $\vec{n}= (n_1, n_2,\cdots , n_N)$ of $N$ identical particles. If the particles are replaced with distinguishable  ones, the possible number of distinctive  arrangements are determined by which particles of number $n_i$ are included in each partition, irrespective of the order. The number of such arrangements are  $\frac{N!}{\prod_{i}(n_i) !}$ (a rigorous proof for this statement is given in Appendix \ref{young} using the Young diagram representation of particle distributions), which corresponds to $Q_{\vec{n}}$.  Since $\prod_i (n_i)!$ is strictly Schur convex \cite{olkin2}, we see that $\vec{n} \prec \vec{m}$ results in $Q_{\vec{n} } $ $>$ $Q_{\vec{m}}$.
\end{proof}

Theorem \ref{1} provides a standpoint to understand the classical hardness of simulating boson sampling with the majorization of input states before unitary operations. 
The physical implication of the theorem is that \emph{a set of $N$ identical particles with a distribution vector $\vec{n}$ is a package of $\frac{N!}{\prod_{i}(n_i) !}$ classical states of the same distribution $\vec{n}$}. This package is delivered to the linear optical network and unpacked in the detector. And we can predict that the hardness of the sampling problem lies in the hardness of unpacking (a simplest example to understand this viewpoint is given in Appendix \ref{example}). 
The decrease of $Q_{\vec{n}}$ along the majorization of $\vec{n}$ means the decrease of the number of classical states per quantum distribution, which results in a simpler unpacking. Thus the lesson of Theorem \ref{1} is that \emph{the more evenly $N$ photons are distributed, the more strongly the indistinguishability of photons adds complexity to the system.} The analysis of later sections includes the specific evidence of this statement.

Theorem \ref{1} also explains the underlying principle of the computational complexity according to the photon density in the modes. 
First, when the photon density is very low ($N \ll M$), the BS process is classically hard to simulate \cite{Aaronson2011}. 
In this limit, the photon population of each mode becomes not larger than one (equivalent to the minimal majorization of photons): it results in the maximal $Q_\vn$.
Second, when the boson density is high ($N \gtrsim M$), the semiclassical samplings (keeping only $\cO(N^{-1})$ terms of the probability amplitude) are efficient \cite{urbina2016,shchesnovich2013}.  This implies that the highly majorized photon distributions under this condition have small $Q_\vn$.

\subsection*{Boltzmann entropy of elementary quantum complexity}

We can find that the former analysis is closely related to the Gibbs ensemble (see, e.g., 5.3 of \cite{bowley1999}). 
It consists of $N$ distinguishable systems, each of which is in one of the quantum states $\p_{i}$ ($1\le i \le M$). With $n_i$ systems in  each $\p_i$ state ($\sum_i n_i =N$), the number of arrangements is given by $N!/ \prod_{i=1}^M (n_i)!$, which is equal to $Q_{\vec{n}}$. Here, $N$ systems correspond to $N$ classical particles and $M$ quantum states to $M$ modes in LON. From this analogy, we define \emph{the Boltzmann entropy of elementary quantum complexity}:

	

\begin{definition}\label{SBQ} The Boltzmann entropy of elementary quantum complexity $S_B(\vec{n})$ for a distribution vector $\vec{n}$ is defined as
	\begin{align}
    \label{boltzmann}
	S_{B}^q(\vec{n}) \equiv \log_2 Q_{\vec{n}} = \log_2 (N!) -\sum_{i=1}^M \log_2 (n_i !),
	\end{align}
	which measures the amount of distinct classical arrangements for the given quantum distribution.   
\end{definition}
Additionally, for later convenience, we also define \emph{the Shannon entropy of elementary quantum complexity} $H$ here:
\begin{definition} The Shannon entropy of elementary quantum complexity $H(\vec{n})$ for a distribution vector $\vec{n}$ is defined as
	\begin{align}
    \label{shannon}
	H(\vec{n}) \equiv \Big( -\sum_{i=1}^M \frac{n_i}{N}\log_2 \frac{n_i}{N}\Big).
	\end{align}
\end{definition}
Regarding the standard interpretation of Shannon entropy as a measure of unpredictability (unknown average information content) for a given physical system (see, e.g., Chapter 1 of \cite{timpson2013}), we can state that $H(\vn)$ in LON measures the unknown information of the photon distribution per mode.
One can see that 
$S_B^q$ goes to $NH$  when all nonzero $n_i$ become large, using the Stirling approximation.

These two entropies are used in Section \ref{e} to analyze the additive error for the approximated computation of transition amplitudes.
\\

\paragraph*{Remark.--} 
So far, we have compared completely indistinguishable particles with completely distinguishable ones. On the other hand, in the realistic setup with imperfect detectors, one should consider the effect of partial distinguishablity \cite{shchesnovich2015, tichy2015} which is the intermediate transition between quantum and classical particles.
The partial distinguishablity of particles is represented with a distinguishabilty matrix, the elements of which can vary continously.     
Hence, it is possible to analyze the behavior of elementary quantum complexity in boson sampling along the two variables, distinguishability (continuous) and majorization (discrete). The role of partial distinguishability in BS when one real parameter can determine the distinguishability is analyzed in Ref. \cite{renema2017efficient}, and we push in this direction further \cite{chinpartial} from the viewpoint of coherence resource theory.

\section{Classical algorithm and majorization}\label{Rt}


From the perspective of Theorem \ref{1}, we expect that the majorization of input and output distribution affects the simulation process of passive LON. In this section we show that exact classical algorithms with shorter runtime can be found when systems have more majorized input and output photon distributions.

For single photon systems (no more than one photon per mode both in input/output states) the best algorithm to simulate the scattering amplitude is given by Ryser's formula \cite{Ryser1963}. It requires an exponential runtime, which scales as $\mathcal{O}(2^N N^2)$ ($N$ is the photon number for the LON case). 
This implies that the classical computation of permanent demands exponential classical resources as the problem size grows. On the other hand,
when input/output photon distribution vectors are arbitary, the generalized Ryser formula can be found (Refs.~\cite{aaronson2012,huh2016computing,yung2016}).

In this section, we show that the behavior of the runtimes for the generalized permanent computing algorithms introduced in Refs.~\cite{aaronson2012,huh2016computing,yung2016} 
follows our prediction. We discuss their relation with the elementary quantum complexity defined in Section \ref{Q}. The linear operation we deal with in this work can be represented as a nontrivial unitary matrix $U$. 
\\

First we briefly introduce the formulae of Ryser's \cite{Ryser1963} and Glynn's \cite{Glynn2010,Glynn2013} for $N\times N$ matrix permanent computation and their generalized formulae \cite{aaronson2012,huh2016computing,yung2016} for the permanents of matrices that have repeated rows and columns. 

The Ryser's formula \cite{Ryser1963} for $N\times N$ matrix $U$ is given by
\begin{align}
\textrm{Per}(U) = \sum_{\vec{x} \in \{0,1\}^N}(-1)^{\sum_{i=1}^N x_i} \prod_{i=j}^N \Big(\sum_{k=1}^N U_{jk}x_{k} \Big),
\end{align}
where $\vec{x} = (x_1,x_2,\ldots x_N)$.
Glynn's formula \cite{Glynn2010,Glynn2013} with the random variable expectation  is given by
\begin{align}
\textrm{Per}(U) = \frac{1}{2^{N-1}}\sum_{\vec{x}\in \{-1,1\}^N}(\prod_{i=1}^N x_i)\prod_{j=1}^N\Big(\sum_{k=1}^N U_{jk}x_k \Big).
\label{Glynn's}
\end{align}
Now the summation of $\vec{x}$ is over $\vec{x} \in \{ -1,1\}^N$, or $\vec{x} \in \mathcal{X} \equiv \mathcal{R}[2]\times \cdots \times \mathcal{R}[2]$, where $\mathcal{R}[i]$ is a set that consists of the $i$th root of unity. The runtime of both formulae for exactly computing the permanent is equal and $\cO(2^N N^2)$ (the improvement to $\cO(2^N N)$ is possible by using Gray code order).  	 

When $U$ has repeated rows $or$ columns, and the $i$th column (or row) is repeated $n_i$-times, which correspond to the case when the input photon distrition vector is $\vn$,  Eq. \eqref{Glynn's} is generalized to \cite{aaronson2012} 
\begin{align}
\textrm{Per}(U) =\sum_{\vec{z}\in \mathcal{X}} v_{\vec{n}}^2(\prod_{i=1}^{N} \bar{z}_i^{n_i}) \prod_{j=1}^N \Big( \sum_{k=1}^N U_{jk}z_k\Big),
\label{genGlynn's}
\end{align}
where $\mathcal{X} \equiv \mathcal{R}[n_1+1]\times \cdots \times \mathcal{R}[n_N+1]$ and $v_{\vn}\equiv \sqrt{\prod_{i=1}^N(n_i!/n_i^{n_i})}$. The classical runtime for the above algorithm is 
\begin{align}
\cO\big[\prod_{k=1}^N(n_k+1)\a_{\vec{n}}N\big],
\label{ahrt}
\end{align}
where $\a_{\vec{n}}$ is defined in \eqref{concave}.

For the more general case when $U$ has repeated rows $and$ columns, and the $i$th column is repeated $n_i$-times and $j$th column is repeated $m_j$-times, which correspond to the case when the input/output photon distrition vectors are $\vn$ and $\vm$,  there exist two algorithms known so far.
One of them is \cite{yung2016}
\begin{align}
\textrm{Per}(U) =\sum_{\vec{z}\in \mathcal{X}} v_{\vec{n}}^2(\prod_{i=1}^{N} \bar{z}_i^{n_i}) \prod_{j=1}^N \Big( \sum_{k=1}^N U_{jk}z_k\Big)^{m_j},
\label{mgGlynn's}
\end{align} 
where $\mathcal{X}$ and $v_{\vn}$ are the same with those in Eq. \eqref{genGlynn's}. Note that the role of $\vm$ reveals from the order of the second parenthesis. This is the direct generalization of \eqref{Glynn's} and \eqref{genGlynn's}.
The other is \cite{huh2016computing}
\begin{align}
\mathrm{Per}\left(U\right)=&
\frac{1}{2^{N}}\sum_{v_{1}=0}^{n_{1}}\cdots\sum_{v_{M}=0}^{n_{M}}(-1)^{N_{v}}
\begin{pmatrix}
n_{1}\\v_{1}
\end{pmatrix}
\cdots
\begin{pmatrix}
n_{M}\\v_{M}
\end{pmatrix} \nn \\
& \qquad\qquad\qquad \times 
\prod_{k=1}^{M}[\sum_{j=1}^M(n_j-2v_j) U_{jk}]^{m_{k}},
\label{eq:permseries2}
\end{align}
which is obtained using a multivariate series expansion introduced by Kan \cite{kan:2007}.

The minimal runtime for the algorithms \eqref{mgGlynn's} and \eqref{eq:permseries2}  are equal, which is counted in Ref. \cite{huh2016computing}. The minimal runtime, denoted as $\cT_{min}(\vec{n},\vec{m})$, is given by
\begin{align}
\cT_{min}(\vec{n},\vec{m})=
\mathcal{O}\Big[ \min\Big( \prod_{i=1}^{M}(n_i+1), \prod_{j=1}^{M}(m_j+1) \Big) \a_{\vec{n} } \a_{\vec{m} }  \Big],
\label{rt} 
\end{align}
where $\min(a,b)$ is the smaller value between $a$ and $b$, and ($\vn$, $\vm$) are distribution vectors defined as in Section \ref{Q}, which holds by the relation between the matrix permanent and transition amplitude of LON.
A special case is when $\vm = \s(\vec{1}_N, \vec{0}_{M-N})$, i.e., $m_i$ are not not greater than $1$ for all $i$ (here $\s(\vec{v})$ represents any permutation vector of $\vec{v}$).   For the case, since $\min\Big(\prod_{i=1}^{M}(n_i+1), \prod_{j=1}^{M}(m_j+1) \Big) = \prod_{i=1}^N (n_i+1)$ and $\a_{\vm}=N$, it is straightforward to see that $\cT_{min}(\vn,\vm)$ is equal to \eqref{ahrt}. A more special case is when  $\vn =  \vm = \s(\vec{1}_N, \vec{0}_{M-N})$, i.e., both $n_i$ and $m_i$ are not greater than $1$ for all $i$. Then $\cT_{min}(\vec{n},\vec{m}) = 2^N N^2$, which is computationally as expensive as that of Ryser's and Glynn's formulae.\\

The relation of majorization with $\cT_{\min}(\vec{n}, \vec{m})$ is clear from the following identity \cite{huh2016computing}:
\begin{align}
\prod_{i=1}^{M}(n_i+1) = \sum_{k=0}^{\a_{\vec{n}} }X_k(\vec{n}) \qquad (X_0\equiv 1)
\end{align}
where $X_k(\vec{n})$ are the elementary symmetric polynomials, defined in Eq. \eqref{concave}. Then Eq. \eqref{rt} is reexpressed as
\begin{align}
\cT_{\min}(\vn,\vm) = \cO\Big[\min\Big(\sum_{k=1}^{\a_\vn} X_{k}(\vn),  \sum_{k=1}^{\a_\vn} X_{k}(\vm) \Big)\a_\vn\a_{\vm} \Big].
\label{rtr}
\end{align}
 Since $\a_{\vec{n} }$ is a Schur-concave function (Eq.\eqref{concave}) and $\sum_{k=0}^{\a_{\vn}} X_{k}(\vn)
$ is a strictly Shur-concave function (this is clear from the fact that $X_{k}(\vn)$ is strictly Schur-concave for $k\ge 2$ \cite{olkin2}), one can notice that the scale of $\cT(\vec{n},\vec{m})$ varies along the majorization of the input/output distributions. \emph{The runtime of two systems with different photon distributions can be compared by examining their degrees of majorization.} In Ref.~\cite{huh2016computing}, the summation $\sum_{k=0}^{\a_{\vn}} X_{k}(\vn)$ is defined as \emph{the Fock state concurrence sum}, which is from the functional relation between the elementary symmetric polynomials and the generalized concurrence \cite{gour2005,chin2017gcc}.

We first consider the algorithm by Aaronson and Hance when $\vn$ or $\vm$ is equal to $\s(\vec{1}_{N}, \vec{0}_{M-N})$. Let us choose $\vm =\s(\vec{1}_{N}, \vec{0}_{M-N})$. Then the runtime \eqref{ahrt} is reexpressed as
\begin{align}
\cT_{\min}(\vn, \vm)  = \cO \Big[(\sum_{k=0}^{\a_{\vn} } X_k)\a_{\vn}N   \Big].
\end{align}
Hence, for two different input distribution $\vn_1$ and $\vn_2$ for this case, we have the following inequality:
\begin{align}
 \vec{n}_1 \prec \vec{n}_2  \Longrightarrow \cT(\vec{n}_1,\sigma(\vec{n}_1) ) > \cT(\vec{n}_2,\sigma(\vec{n}_2) ) 
\end{align}
by the strict Schur concavity of $\cT_{\min}(\vn, \vm).$

A slightly more general (but equaly simple) case is when  $\vec{m}$ is a permutation vector of $\vec{n}$ (equally majorized), i.e., $\vec{m} = \sigma(\vec{n})$. Then we have
\begin{align}
\label{rtperm}
\cT(\vec{n},\sigma(\vec{n}))_{\min}=
\mathcal{O}\Big[ (\sum_{k=0}^{\a_{\vn}}X_{k}(\vn)) \a_{\vec{n} }^2 \Big],
\end{align}
and
\begin{align}
 \vec{n}_1 \prec \vec{n}_2  \Longrightarrow \cT_{\min}(\vec{n}_1,\sigma(\vec{n}_1) ) > \cT_{\min}(\vec{n}_2,\sigma(\vec{n}_2) ).
\end{align}

Next, if $\vec{n} \prec \vec{m}$ (the input distribution is majorized by the output distribution), the runtime is given by
\begin{align}
\cT_{\min}(\vec{n},\vec{m}) = \mathcal{O}\Big[ \Big(\sum_{k=1}^{\a_{\vm}}X_{k}(\vm)\Big)  \a_{\vec{n} } \a_{\vec{m} }  \Big].
\end{align}

On the other hand, as we can see in Eq. \eqref{rtr}, if $\vec{n} \succ \vec{m}$, the role of two distribution vectors are reversed. With this input/output symmetry, which comes from the unitarity of the linear optical network operations, the runtime is denoted hereafter by
\begin{align}
\label{rtredef}
\cT_{\min}(\vec{n} \preceq \vec{m}) \equiv \mathcal{O}\Big[\Big( \sum_{k=0}^{\a_{\vec{m} } } X_k(\vec{m}) \Big) \a_{\vec{n} } \a_{\vec{m}} \Big],
\end{align} 
 where the expression $\vec{n} \preceq \vec{m}$ means that $\vec{n}$ is always majorized by $\vec{m}$ or a permutation vector of $\vec{m}$ ($\vm$ is not majorized by $\vn$ in other words), irrespective of the input/output distinction.

With Eq. \eqref{rtredef}  
we can compare two sets of input/output photon distributions, $(\vec{n}_1\preceq \vec{m}_1)$ and $(\vec{n}_2\preceq \vec{m}_2)$.
The inequalities given by possible majorization relations between  $(\vec{n}_1 \preceq \vec{m}_1)$ and $(\vec{n}_2 \preceq \vec{m}_2)$ are listed in the table below. 
\begin{center}
	\begin{tabular}{|  p{3.8cm} | l |}
		\hline
 Majorization relation &  Runtime relation	 \\ \hline
 $\vec{n}_1 = \sigma(\vec{n}_2),$ $\vec{m}_1 = \sigma(\vec{m}_2)$ &
		$\cT_{\min}(\vec{n}_1\preceq \vec{m}_1) =\cT_{\min}(\vec{n}_2\preceq \vec{m}_2)$  \\ \hline
		$\vec{n}_1 \preceq \vec{m}_1 \preceq \vec{n}_2 \preceq \vec{m}_2$ & $\cT_{\min}(\vec{n}_1\preceq \vec{m}_1) \ge \cT_{\min}(\vec{n}_2\preceq \vec{m}_2)$
        \\ \hline
		$\vec{n}_1\preceq \vec{n}_2 \preceq \vec{m}_1 \preceq \vec{m}_2$ &  $\cT_{\min}(\vec{n}_1\preceq \vec{m}_1) \ge \cT_{\min}(\vec{n}_2\preceq \vec{m}_2)$ \\ \hline 
 $\vec{n}_1\preceq \vec{n}_2 \preceq \vec{m}_2 \preceq \vec{m}_1$ & Not determined \\ \hline
		
	\end{tabular}
    
  \end{center}
    
The table shows that if both $\vn_2$ and $\vm_2$ are not majorized by at least one of $(\vn_1, \vm_1)$, then the inequality $\cT(\vec{n}_1\preceq \vec{m}_1) \ge \cT(\vec{n}_2\preceq \vec{m}_2)$ holds (the first three cases in the table support it). 
\\


We can intuitively understand the behavior of $\cT(\vn \prec \vm)$ along the majorization of $(\vn,\vm)$
by considering the physical implication of $S_B^q$. 
Dissecting the functional components of $\cT(\vn \prec \vm)$, $\prod_i(m_i+1)$ arises from the summations from $v_i=0$ to $v_i=n_i$ ($1 \le i \le M$) in Eq. \eq{eq:permseries2}; 
conducting summations for each mode is equivalent to counting the identical photons one at a unit time. In this procedure, the identical particles become distinct (or effectively distinguishable) with each other. In other words, \emph{we treat the quantum particles classically in the classical algorithm of computing matrix permanents}; at the detector, the procedure can be considered as the ``unpacking process'' that we mentioned earlier in Theorem \ref{1} in Section \ref{Q} (a simple example that intuitively explains our current statement is presented in Appendix \ref{example}). 

\begin{figure}[t]
	\centering
	\includegraphics[width=8cm]{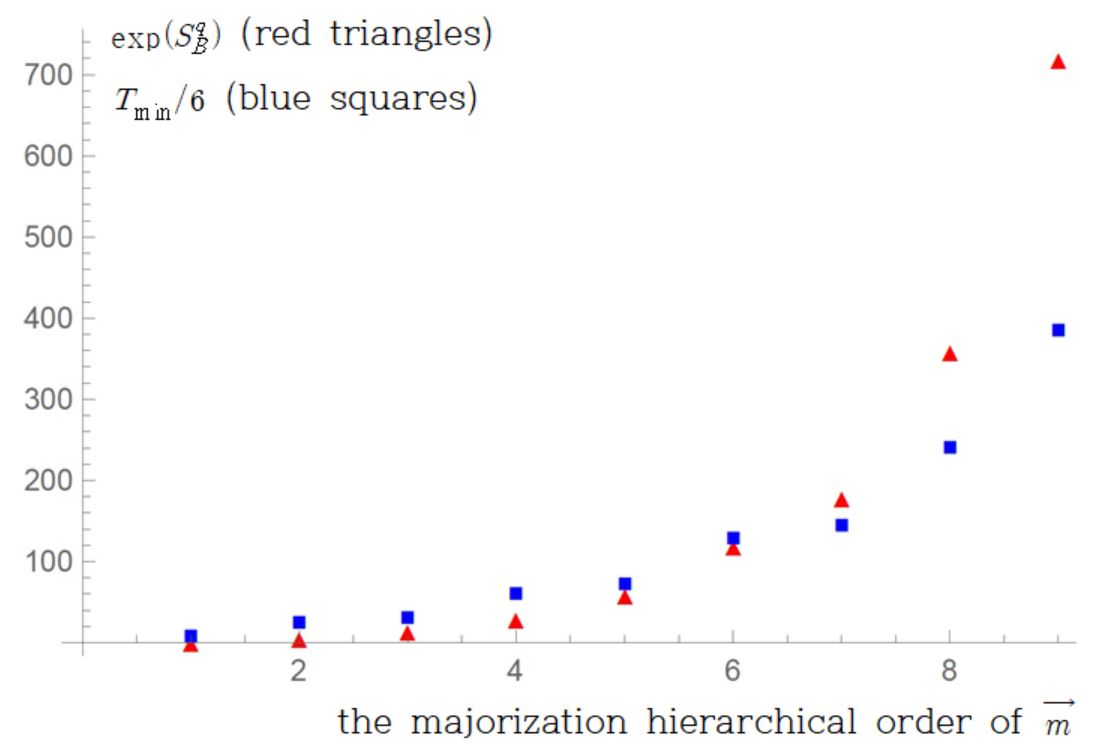}
	\caption{$\exp(S_B^q)= Q_{\vm}$ and $T_{\min}(\vn \prec\vm)/6$ along the hierarchical order of $\vm$ with $N=6$. The horizontal axis represents the discrete majorization hierarchy of $N=6$ as follows: $\vm = (6, \vec{0}_{M-1}) \to (5,1,\vec{0}_{M-1} ) \to (4,2,\vec{0}_{M-1}) \to (4,1,1,\vec{0}_{M-1}) \to (3,2,1,\vec{0}_{M-1})\to (3,1,1,1,\vec{0}_{M-1}) \to (2,2,1,1,\vec{0}_{M-1}) \to (2,1,1,1,1,\vec{0}_{M-1})\to (1,1,1,1,1,1,\vec{0}_{M-1})$. The vertical axis represents $\exp(S_B^q)= Q_{\vm}$ (red triangles) and $T_{\min}(\vn \prec\vm)/6$ (blue squares) for each $\vm$. $Q_{\vm}$ and $T_{\min}(\vn \prec\vm)/6$ are chosen so that the two quantities can be compared in the same graph. Note that $\vn$ is fixed to $(1,1,1,1,1,1,\vec{0}_{M-1})$. One can see that as $\vm$ becomes less majorized from $1\equiv (6, \vec{0}_{M-1})$ to $9\equiv (1,1,1,1,1,1,\vec{0}_{M-1})$, both Boltzmann entropy and runtime increase monotonically as expected.}
	\label{N=6}
\end{figure}

This standpoint discloses the relation of $\big(S_B^q(\vec{n}),S_B^q(\vm)\big)$ with $\cT_{min}(\vn \prec \vm)$.
As $\vn$ and $\vm$ are more majorized, $S_B^q(\vec{n})$ and $S_B^q(\vm)$ become smaller by Shur-concavity.
$S_B^q$ measures the number of effective distinguishable states per quantum distribution (Definition \ref{SBQ}), and the states are to be counted classically. In consequence, the runtime for counting photons decreases as $S_B^q$ decreases (an $N=6$ example is presented in Fig. \ref{N=6}. The graph shows the relation between $\cT_{\min}$ and $S_B^q$ for the case when $\vn$ is fixed to $(1,1,1,1,1,1, \vec{0}_{M-6})$ ($M \ge 6$) and $\vm$ changes from $(6, \vec{0}_{M=1})$ to $(1,1,1,1,1,1, \vec{0}_{M-6})$ as noted under Fig. 4).

\section{Additive error bound and majorization}\label{e}

Now we focus on the approximated algorithm for the generalized matrix permanent. We delve into the relation of majorization with the runtime $\cT_A$ the additive error bound $\cE$ for the approximated permanent computation. While the approximate runtime is again determined with the majorization of input/output photon distributions, the error bound is determined by the majorization difference (the difference of degree of majorization for two discrete vectors, defined in Appendix) of them. And it is exactly expressed with the Boltzmann and Shannon entropy of elementary quantum  complexity defined in Section \ref{Q}.

The approxmated algorithm that computes the matrix permanent in polynomial time was first suggested by Gurvits \cite{gurvits2005}. The algorithm has the runtime  $\cT_A= \cO(N^2/\eps^2)$ with the nontrivial additive error bound $\cE=\eps$. The case when the computed matrix has repeated rows or columns was presented in \cite{aaronson2012}. Supposing the input $\vn$ is distributed arbitrary, the algorithm has the runtime $\cO(N\a_{\vn}/\eps^2)$ with the error bound $\cE=\eps v_{\vn}$ ($v_{\vec{n}}\equiv \prod_{i=1}^{M}\sqrt{\frac{n_i !}{n_i^{n_i}} }$).
The more general case when the computed matrix has repeated rows and columns was presented in \cite{yung2016}, which has   the runtime $\cO(\a_{\vn}\a_{\vm}/\eps^2)$ with the error bound $\cE=\eps\cdot\min \Big[ \frac{v_{\vec{n}}}{v_{\vec{m}}},  \frac{v_{\vec{m}}}{v_{\vec{n}}}  \Big]$.  This result includes the result of \cite{aaronson2012}. The algorithms we mentioned so far are compared in the table below.

\begin{center}
	\begin{tabular}{| l | l | l |}
		\hline
		 Input/output & runtime $\cT_A$ & $\cE$ \\ \hline
		 $\vn = \vm =\hat{\sigma}(1,\cdots, 1, 0,\ldots,0)$  &  $\cO(N^2/\eps^2)$ & $\eps$ 	 \\ \hline
		 Input $or$ output is arbitrary & $\cO(N\a_{\vec{m}}/\eps^2)$ & $\eps v_{\vm}$ \\ \hline 
	Input $and$ output are arbitrary & $\cO(\a_{\vec{n}}\a_{\vec{m}}/\eps^2)$ & $\eps \frac{v_{\vm} }{v_\vn}$ \\ \hline
		
	\end{tabular}
\end{center} 


A first direct observation is that the approximated runtime decreases as the degrees of majorization for input/output distribution vectors increase, which is clear since $\a_{\vn}$ is Schur-concave. Therefore, we can apply the same analysis on the exact runtime $\cT_{\min}$ to the approximated runtime $\cT_A$.

On the other hand, the behavior of the error bound $\cE$  along with the majorization change is more interesting and reveals its close relation with the Boltzmann entropy of elementary quantum complexity.
Since $v_{\vec{n}}\equiv \prod_{i}\sqrt{\frac{n_i !}{n_i^{n_i}} }$ is another Schur-concave function (Eq. \eqref{concave}), similar to Eq. \eqref{rtredef}, we redefine the error bound as
\begin{align}
\mathcal{E}_{\vec{n} \preceq \vec{m}} \equiv  \eps\cdot \frac{v_{\vec{m}}}{v_{\vec{n}} }.
\end{align}
The above equation shows that the error bound depends on the \emph{majorization difference} between input/output distributions, which quantifies how much one discrete vector is more majorized than the other (the rigorous definition of majorization difference is presented in Appendix, Definition \ref{diff}). A greater majorization difference corresponds to a smaller additive error bound.

For example,
when $\vec{m}= \sigma(\vec{n})$ (no majorization difference), $v_{\vec{n}} = v_{\vec{m} }$ and the error bound becomes maximal ($=\eps$). As the other extreme example, 
when $\vec{n} = (1,1,\cdots , 1)$ and $\vec{m} = (N, \vec{0}_{N-1})$ (the largest majorization difference), we have $(v_{\vec{n}} , v_{\vec{m}}) = (1,N!/N^N)$. Then as $N$ becomes very large ($\gg 1$), the error bound becomes very small ($ \to \eps \sqrt{2\pi N}e^{-N}$).

We can also compare the error bound of two systems with different input/output states $( \vec{n}_1 \preceq \vec{m}_1)$ and  $( \vec{n}_2 \preceq \vec{m}_2)$ for some cases as follows:
\begin{align}
\vec{n}_1 =\sigma(\vec{n}_2) \preceq \vec{m}_1 \preceq \vec{m}_2 \quad\Longrightarrow\quad  \mathcal{E}_{\vec{n}_1 \preceq \vec{m}_1} \ge \mathcal{E}_{\vec{n}_2 \preceq \vec{m}_2}, \nn \\
\vec{n}_1 \preceq \vec{n}_2 \preceq \vec{m}_1 = \sigma(\vec{m}_2) \quad\Longrightarrow\quad  \mathcal{E}_{\vec{n}_1 \preceq \vec{m}_1} \le \mathcal{E}_{\vec{n}_2 \preceq \vec{m}_2}.  
\end{align} 
Both inequalities are consistent with the statement that the additive error decreases as the majorization difference increases.

The additive error bound $\cE_{\vec{n}\preceq \vec{m} }$ has a close functional relation with the entropies of elementary quantum complexity defined in Section \ref{Q}.
By taking the logarithm, we find that $\cE_{\vec{n}\preceq \vec{m}}$ is expressed with the difference between $S_B^q$ and $H^q$:
\begin{align}
2\log_2\Big[\frac{\mathcal{E}_{\vec{n} \preceq \vec{m}}}{\eps} \Big] &=\log_2 (v_{\vec{m}})^2 - \log_2 (v_{\vec{n}})^2 \nn \\
& = \big(\log_2 N! - \log_2 [ \prod_i (n_i)! ]  +\sum_i n_i\log_2 n_i \big)  \nn \\
& \quad  - \big( n_i \leftrightarrow m_i \big) \nn \\
& = \big(S_B^q(\vec{n}) -NH^q(\vec{n}) ) - (S_B^q(\vec{m}) -N H^q(\vec{m}) \big), 
\end{align}
which gives
\begin{align}
\label{esbsg}
\mathcal{E}_{\vec{n} \prec \vec{m}} = \eps\cdot 2^{\frac{1}{2}\big( \Delta S(\vec{n})-\Delta S(\vec{m}) \big)},
\end{align}
where $\Delta S (\vec{n}) \equiv S_B^q( \vec{n}) -NH^q(\vec{n} )$.
Note that $S_B^q( \vec{n})$, $H^q(\vec{n} )$ and $\Delta S(\vec{n})$ are  all Schur-concave functions. 

Eq. \eqref{esbsg} indicates an intriguing property of LON: rather than the Boltzmann (or Shannon) entropy itself, \emph{the difference between the Boltzmann and Shannon entropy, $\Delta S$, determines the additive error}. When all nonzero elements of $\vn$ becomes large, i.e., $n_i \gg 1$, $\Delta S$ goes to zero by Stirling approximation.

We can explain the implication of Eq. \eqref{esbsg} in consideration of the physical process related to the error bound.
The additive error is calculated for fixed boundary conditions, i.e., some specific input and output photon distributions. By preparing a state before the interferometic operation (pre-selecting a state), we obtain the information about the input photon distribution. As mentioned under Eq. \eqref{shannon}, the obtained information per mode is quantified by Shannon entropy $H$. As a result, $NH$ should be subtracted when the error for the approximated transition probability is considered. The same interpretation also works for the output state (post-selected state).
Since the information on the photon distribution is known to us by post-selection,  the corresponding quantity $NH$ should be subtracted as well.


We can also see the behavior of $\cE_{\vec{n}\preceq \vec{m} } $ when the photon number is very large (approaching the Bose-Einstein condensation limit) using Stirling's formula
\begin{align}
-\log_2 n! +n\log_2 n = \log_2 e\cdot n -\frac{1}{2}\log_2 2\pi n + \cO(\frac{1}{n}).  
\end{align}
Then
\begin{align}
\Delta S(\vec{n}) \simeq  \log_2 e \cdot N -\frac{1}{2} \sum_{i=1}^M \log_2 2\pi n_i,   
\end{align}
when all $n_i$ are large, and 
\begin{align}
\label{ebec}
\cE_{\vec{n}\preceq \vec{m} } =\eps\cdot 2^{-\frac{1}{4}\sum_i \log_2 \frac{n_i}{m_i} }= \eps \prod_{i=1}^M \Big(\frac{m_i}{n_i}\Big)^\frac{1}{4}. 
\end{align}
With the strict Schur-concavity of $\sum_i\log_2 n_i$ \cite{olkin2}, it is directly seen that the most right hand side of Eq. \eqref{ebec} is not larger than $\eps$.

Before closing this section, it is worth mentioning that the additive error bound analysis we presented here is hard to detect when the transition amplitudes themselves are very small. For such cases, multiplicative errors should be considered instead of additive errors. It would be another interesting topic whether some pattern that depends on the majorization again appears in the multiplicative error analysis.

\paragraph*{Remark.--}   It is well-known that a LON with nonclassical photon input can generate entanglement (see, e.g., \cite{kim2002ms, wang2002xb, jiang2013z, stanisic2017generating}). Considering entanglement is a crucial quantum resource in several quantum computational processes, one can ask whether $\cT$ and $\cE$ simulating LON scattering amplitudes have some relation with entanglement. A quantitative analysis on the generation of entanglement in LON is recently given in Ref.~\cite{stanisic2017generating}. We can compare the results in there with ours. However, the behavior of the generated entanglement looks different from the patterns of $\cT$ and $\cE$ (see Table I of Ref.~\cite{stanisic2017generating}).  While $\cT$ decreases when photons are bunched, the entanglement bound for bunched photons are higher than the unbunched case. Since the amount of entanglement depends on the partition of the modes and the choice of the operator $\hU$, it looks obscure to find some manifest relation of the generated entanglement with $\cT$ and $\cE$. On the other hand, as we have discussed so far, the role of majorization (and $S_B^q$) determined by the input/output photon distributions is obvious. 




\section{Discussions}\label{discussion}

In this paper, we presented quantitative relations between majorization and classical simulations of LON scattering processes. As a measure of the elementary complexity of a given photon distribution, we introduced a Schur-concave function $S_B^q$ (the Boltzmann entropy of elementary quantum complexity). The generalized classical runtime for computing the matrix permanent $\cT$ and the additive error $\cE$ for approximated permanent computation were analyzed with the majorization and $S_B^q$. For systems with more majorized input/output photon distributions, the known classical algorithms to simulate the LON scattering process had shorter runtime.

We expect that our current research can develop in several directions. For example, the direct relation between $\cE$ and the Boltzmann entropy can be compared with the relation between $\cT$ and the generalized concurrence explained in \cite{huh2016computing}, which can provide a clue to understand the quantum supremacy of LON from the perspective of resource theories with computable quantum measures (see, e.g., \cite{plenio2005introduction,streltsov2017colloquium}). Our results also could be applied to the computational complexity problems of linear optical systems with continuous variables~\cite{rahimi2015,rahimi2016sufficient,hamilton2017gaussian}.

\section*{Acknowledgements}
The authors are grateful to the anonymous referees for their helpful advice. 
This work was supported by Basic Science Research Program through the National Research Foundation of Korea (NRF) funded by the Ministry of Education, Science and Technology (NRF-2015R1A6A3A04059773). In addition, 
S. C. acknowledges Professor Jung-Hoon Chun for his generous support during the research.

\appendix

\section{Young diagrams and majorization}\label{young}

A Young diagram is a set of arranged boxes in left-justified rows with non-increasing row lengths from the top. With $N$ boxes in it, a diagram represents a partition of an integer $N$ (this can be understood as a partition of $N$ identical particles into modes for our case). We can construct a hierarchy of Young diagrams with the same $N$ as in the definition  below \cite{ruch}. Then the hierarchy has a one-to-one correspondence with a majorization relation of discrete vectors of $N$.

\begin{definition}
	\label{greater}
	A diagram $\gamma$ is greater than a diagram $\gamma '$ (or $\gamma \supset \gamma'$), if $\gamma$ can be built from $\gamma'$ by moving boxes  upward (from shorter rows to longer or equal-length ones).
\end{definition}
One example is as follows:
	\begin{align}
		\Yvcentermath1
 \gamma = \yng(5,1,1) \quad \supset \gamma' =\yng(3,2,2), \nn 
\end{align}
since we can build $\gamma$ from $\gamma'$ by moving one box in the second row and another in the third row to the first row. 
 
Suppose that a discrete vector $\vec{\gamma}$ is defined as a $N$-dimensional vector with $\gamma_i=$(number of boxes in the $i$'th row).  Then we have
\begin{align}
\vec{\gamma} \succ \vec{\gamma}' \iff \gamma \supset \gamma'.
\end{align}
Hence, the non-increasing distribution vectors of $N$ bosons are exactly mapped to the Young diagrams of $N$ boxes. A more intuitive proof of Theorem \ref{1} comes from this relation as follows:
\begin{proof}[Proof of $Q_{\vn}=N!/\prod_i(n_i)!$ with Young diagrams] A distribution of $N$ identical particles corresponds to a Young diagram with $N$ boxes, and the number of distributions of $N$ distinguishable particles with the same partition corresponds to Young tableaux (Young diagrams with a label on each box). For example, if 7 identical particles are distributed (partitioned) as the $(3,2,2)$ modulo permutation (just as $\gamma'$ of the above example), the correponding arrangements for 7 distinguishable particles with the partition $(3,2,2)$ is given by
\begin{align}
\Yvcentermath1
\young(123,45,67), \quad \young(124,35,67), \quad \young(245,17,36), \cdots .  \nn
\end{align} 
There are $\frac{7!}{3!2!2!}$ possible ways of labeling boxes. For $N$ particles with a distribution vector $\vec{n}=(n_1,n_2,\cdots, n_M)$ with $\sum_{i=1}^M n_i=N$, the number is given by $\frac{N!}{\prod_i(n_i)!}$.
\end{proof}

Definition \ref{greater} provides a complete array of Young diagrams with $N$ boxes, which is equivalent to a complete array of distribution vectors of $N$ identical particles along the majorization. As a simple example, the array of young diagrams with $N=5$ is as follows:
\begin{align}
\tiny\Yvcentermath1
&\yng(5) \supset
 \yng(4,1) \supset
 \yng(3,2) \supset \yng(3,1,1) \nn \\
&\supset \yng(2,2,1) \supset \yng(2,1,1,1) \supset \yng(1,1,1,1,1) \nn \\ 
\end{align}
In the distribution vector form, it is rewritten as
\begin{align}
&(5, \vec{0}_{M-1}) \to (4,1, \vec{0}_{M-2}) \to (3,2, \vec{0}_{M-2}) \nn \\ &\to (3,1,1, \vec{0}_{M-3})  \to (2,2,1, \vec{0}_{M-3})\to (2,1,1,1, \vec{0}_{M-4}) \nn \\
& \to (1,1,1,1,1, \vec{0}_{M-5}).
\end{align}
One arrow in the array correponds to one step of the vector redistribution, and the $N=5$ array consists of 6 steps. 
The majorization difference between two distribution vectors is defined using the array of Young diagrams:
\begin{definition}
	\label{diff}
	For two distribution vectors $\vn \prec \vm$ ($\sum_i n_i = \sum_i m_i =N$) and their corresponding young diagrams $\gamma_n \subset \gamma_m$, the majorization difference of $\vn$ and $\vm$ is equal to the number of steps from  $\gamma_n$ to $\gamma_m$.
\end{definition} 
For example, the majorization difference of $\vn = (4,1, \vec{0}_{M-2})$ and $\vm =  (2,1,1,1, \vec{0}_{M-4})$ is 4. 
\\

\section{Comparison of distinguishable and indistinguishable scattering processes with the simplest example}\label{example}

Here we compare the scattering process of distinguishable and indistinguishable particles in passive LON for the simplest case, i.e., $(N,M)=(2,2)$. Let's suppose that two particles are in different modes to each other.
When the particles are distinguishable, they are represented as different creation operators $a_i^{\dagger}$ and $b_i^{\dagger}$. Then the initial state is, e.g., given by $\ha^{\dagger}_1 \hb^{\dagger}_2|0\>$ (note that there exists another possibility, $\hb_1^{\dagger}\ha_2^{\dagger}|0\>$).  Therefore, under the unitary operations $\hU\ha_i^{\dagger}\hat{U}^{\dagger} = \sum_{j}U_{ij}\ha_j^{\dagger}$ and $\hU\hb_i^{\dagger}\hat{U}^{\dagger} = \sum_{j}U_{ij}\hb_j^{\dagger}$, the final state becomes
	\begin{align}
	\Big[ U_{11}U_{21}\ha_1^{\dagger}\hb_1^{\dagger} + U_{11}U_{22}\ha_1^{\dagger}\hb_2^{\dagger} + U_{12}U_{21}\ha_2^{\dagger}\hb_1^{\dagger} + U_{12}U_{22}\ha_2^{\dagger}\hb_2^{\dagger}\Big] |\vec{0}\>.
	\end{align} The amplitudes of each distinctive state is determined by one term, i.e., $U_{11}U_{21}$, etc.  
On the other hand, when particles are distinguishable, the initial state is given by
	$\ha^{\dagger}_1 \ha^{\dagger}_2|0\>$
and the final state becomes
	\begin{align}
	&\Big[ U_{11}U_{21}(\ha_1^{\dagger})^2 + (U_{11}U_{22} + U_{12}U_{21})\ha_1^{\dagger}\ha_2^{\dagger} + U_{12}U_{22}(\ha_2^{\dagger})^{2}\Big] |\vec{0}\> \nn \\
	&=	\Big[ U_{11}U_{21}(\ha_1^{\dagger})^2 + \textrm{perm}[U]\ha_1^{\dagger}\ha_2^{\dagger} + U_{12}U_{22}(\ha_2^{\dagger})^{2}\Big] |\vec{0}\>.
	\end{align}

When two particles are in the same mode at the output ($\vn = (2,0)$ or $(0,2)$), only one term determines the scattering amplitude. However, when two particles are in different modes at the output,  two terms to determine the scattering amplitude. To  simulate this amplitude classically, we need to add these terms one by one. This is what we meant in Section \ref{Rt} that \emph{we treat the quantum particles classically in the classical algorithm of computing matrix permanents}.
Hence, this example intuitively explains how the (identical) photon number distribution affects the complexity of passive LON process.

\bibliography{MajEntComp}

\end{document}